\documentclass[11pt]{article}
\usepackage{fullpage}
\usepackage[colorlinks=true]{hyperref}
\usepackage{url}
\usepackage{amsmath,amsfonts,amsthm,amssymb,multirow}
\usepackage{mdframed}
\usepackage{graphicx}
\usepackage{algorithm}
\usepackage[noend]{algpseudocode}
\usepackage{bbold}
\usepackage{mathtools} 
\usepackage{soul, color}
\allowdisplaybreaks
\newenvironment{reminder}[1]{\bigskip
	\noindent {\bf Reminder of #1  }\em}{\smallskip}
\usepackage{comment}	

\DeclarePairedDelimiter\abs{\lvert}{\rvert}%

\def \poly {\text{poly}}
\def \exp {\text{exp}}
\def \polylog {\text{polylog}\ }
\def \eps {{\varepsilon}}

\def \NP {\mathsf{NP}}

\newtheorem{lemma}{Lemma}[section]
\newtheorem{definition}{Definition}[section]

\newtheorem{corollary}{Corollary}[section]
\newtheorem{theorem}{Theorem}[section]

\newtheorem{conjecture}{Conjecture}

\newtheorem{open}{Open Problem}

\newcommand{\Thr}[1]{\text{Thr}_{#1}}
\newcommand{\Bool}[1]{\{0,1\}^{#1}}
\newcommand{\PCP}{\text{PCP}}

\makeatletter
\let\c@fconjecture\c@conjecture
\makeatother

\makeatletter
\let\c@fconj\c@conj
\makeatother

\title{Imperfect Gaps in Gap-ETH and PCPs}
\author{Mitali Bafna\footnote{Harvard University, mitalibafna@g.harvard.edu, Supported by NSF Grant CCF-1565641 and CCF-1715187}, Nikhil Vyas\footnote{MIT, nikhilv@mit.edu, Supported by NSF Grant CCF-1741615}}
\date{}

\begin{document}
\maketitle

\begin{abstract}
We study the role of perfect completeness in probabilistically checkable proof systems (PCPs) and give a new way to transform a PCP with imperfect completeness to a PCP with perfect completeness, when the initial gap is a constant. In particular, we show that $\PCP_{c,s}[r,q] \subseteq \PCP_{1,1-\Omega(1)}[r+O(1),q+O(r)]$, for $c-s=\Omega(1)$. This implies that one can convert imperfect completeness to perfect in linear-sized PCPs for $NTIME[O(n)]$ with a $O(\log n)$ additive loss in the query complexity $q$. We show our result by constructing a ``robust circuit'' using threshold gates. These results are a gap amplification procedure for PCPs (when completeness is imperfect), analogous to questions studied in parallel repetition~\cite{Rao-PR} and pseudorandomness~\cite{chernoff-exp}. 

We also investigate the time complexity of approximating perfectly satisfiable instances of 3SAT versus those with imperfect completeness. We show that the Gap-ETH conjecture without perfect completeness is equivalent to Gap-ETH with perfect completeness, i.e. we show that Gap-3SAT, where the gap is not around 1, has a subexponential algorithm, if and only if, Gap-3SAT with perfect completeness has subexponential algorithms. We also relate the time complexities of these two problems in a more fine-grained way, to show that $T_2(n) \leq T_1(n(\log\log n)^{O(1)})$, where $T_1(n),T_2(n)$ denote the randomized time-complexity of approximating MAX 3SAT with perfect and imperfect completeness, respectively.
\end{abstract}

\section{Introduction}

The PCP theorem~\cite{ALMSS} was a breakthrough result showing that $\NP$ has proofs verifiable using only $O(1)$ bits and constant probability of error, with a polynomial blow-up in the size of the proof. The theorem led to a flurry of activity in getting the best set of parameters: the soundness (the probability of acceptance of incorrect proofs), proof size and queries. PCP constructions were instrumental in showing optimal hardness of approximation results for a host of problems such as $k$-SAT and 3-LIN~\cite{Hastad}. Despite this progress, many important questions have remained wide open. For instance: \emph{do there exist linear-size PCPs for $NTIME[O(n)]$, with constant queries and constant soundness?} Hence we believe it is important to understand the role of all the parameters in PCPs, and we focus our attention on the completeness of these proof systems.

We investigate the question: \emph{can imperfect completeness help obtain better PCPs?} The size versus query tradeoff in PCPs has been extensively studied: A long line of work culminated in a PCP for $NTIME[O(n)]$ with $O(n\cdot \polylog n)$ size and $O(1)$ queries~\cite{dinur}. On the other hand, Ben-Sasson et al~\cite{Ben-SassonKKMS13} achieved a linear-sized PCP for $NTIME[O(n)]$ with $O(n^{\epsilon})$ query size for all constants $\epsilon > 0$.\footnote{This particular construction is non-uniform. To our knowledge no explicit PCPs with $o(n)$ query complexity, constant soundness and linear size are known.}
These results are far from what is conjectured: namely, that PCPs exist with $O(1)$ queries and linear size. In this thesis, we show how to transform any PCP with imperfect completeness and constant gap (between soundness and completeness) to one with perfect completeness and a mild additive extra number of queries. The loss in query complexity in the transformation from imperfect to perfect completeness in the latter regime (of linear-size) is inconsequential in comparison to the query complexity of~\cite{Ben-SassonKKMS13}.

Although in current PCP constructions for $NTIME[O(n)]$, perfect completeness might come for free when one does not care about the verifier's predicate, PCPs with imperfect completeness are very important in showing optimal hardness of approximation for problems like 3LIN~\cite{Hastad}, where deciding satisfiability is in polynomial time. For other CSPs like Max 1-in-$k$-SAT one can get substantially better approximation algorithms for perfectly satisfiable instances~\cite{V1ink}. The Unique Games Conjecture of Khot~\cite{Khot} asks for a PCP with ``unique'' queries and imperfect completeness, the latter being necessary due to the tractability of satisfiable instances of Unique Games. Although in some previous cases like 3LIN, imperfect completeness was necessary, but for cases like 2-to-1 games and Max $k$-CSP one would guess that the same hardness of approximation results should hold with perfect completeness. Unfortunately all the current methods~\cite{DKKMS,Chan} incur a loss in completeness, and it is unclear whether this is because of the nature of the problem or due to the inefficacy of current methods. This leads to the central question: \emph{Given a CSP, how hard is it to approximate instances that are perfectly satisfiable, compared to those that are not?}


We also study this question in a fine-grained way, and compare the time complexities of approximating satisfiable versus imperfectly satisfiable instances of 3SAT. $\NP$-hardness results (while very useful in measuring intractability with respect to poly-time algorithms) do not imply tight or even superpolynomial lower bounds for the running time.  The Exponential Time Hypothesis (ETH)~\cite{eth} states that there are no $2^{o(n)}$ time algorithms for deciding satisfiability of 3SAT. Through the equivalence between PCPs and gap problems, using state of the art PCPs~\cite{dinur, Ben-SassonS08} there is a reduction from a 3SAT instance on $n$ variables and clauses to a Gap-3SAT instance with $O(n\cdot \polylog n)$ variables and clauses. This proves that under ETH, Gap-3SAT does not have $O(2^{n/\log^c(n)})$ algorithms for some fixed $c$, whereas Gap-3SAT has eluded even $2^{o(n)}$ algorithms. To get around precisely this gap, the Gap-ETH hypothesis was proposed~\cite{dinur-gap-eth, gap-eth2}. Gap-ETH states that Gap-3SAT does not have $2^{o(n)}$ algorithms. This hypothesis has led to several tight inapproximability results \cite{geth1, geth2, geth3, geth4} with respect to the running time required. We study the role of perfect completeness in Gap-ETH, where Gap-ETH without perfect completeness is the hypothesis that there are no $2^{o(n)}$ algorithms for Gap-3SAT without perfect completeness.

Gap amplification is in itself an important problem studied in the context of parallel repetition~\cite{RazPR}, error reduction and pseudorandomness~\cite{chernoff-exp}. We study this problem in PCPs and show a way to transform any PCP into a one-sided error one. Similar questions of gap amplification when completeness is not 1, have been studied
for parallel repetition~\cite{Rao-PR}, but these results incur a huge blow-up in the alphabet and hence cannot be applied to get perfect completeness in PCPs with constant alphabet. These techniques in parallel repetition have been used in quantum computation, to show instances of multi-player games with large separation between the entangled and classical value and amplification of entangled games~\cite{Rao-PR,AnchorPR}.

\subsection{Our contributions}


\begin{paragraph}{\bf PCPs without perfect completeness:}
We give an efficient way to boost the completeness of any PCP, which makes the completeness 1. Our results go via the construction of ``robust circuits'' for the approximate threshold function on $n$ bits. These circuits are of depth $O(\log n)$, fan-in $O(1)$ and size $O(n)$, and use successive layers of threshold gates to boost the fraction of ones in inputs
that have large Hamming weight, while maintaining the fraction of ones in other inputs below a certain threshold. The circuits are tolerant to some form of adversarial corruptions and this property allows us to prove the soundness of the new PCP. Our main theorem is the following:

\begin{theorem}\label{thm:pcp}
Let $c, s \in (0,1), s < c$ be constants. There exists a constant $s' \in (0,1)$ depending only on $c, s$ such that
$$\PCP_{c, s}[r, q] \subseteq \PCP_{1,s'}[r, q +O_{s,c}(r)].$$ Furthermore, if the original proof size was $n$, then the final proof size is $n+O(2^r)$.
\end{theorem}

Note that in the above theorem, one can prove inclusion in a PCP class, with arbitrary constant $s''$ (instead of a fixed constant $s'$) by applying derandomized serial repetition ($\PCP_{1,s'}[r,q] \subseteq \PCP_{1,s''}[r,O(q)]$ with same proof size). This does not blow up the size of the PCP and the query complexity only increases by a constant factor.

As a corollary we show that linear-sized PCPs for $NTIME[O(n)]$ with $q$ queries and imperfect completeness can be converted to a linear-sized PCPs with perfect completeness and $q+O(\log n)$ queries. Current PCP constructions with constant rate and alphabet have query complexity $n^{\Omega(1)}$~\cite{Ben-SassonKKMS13}, so we show that to improve upon this, it is enough to construct linear sized PCPs with imperfect completeness and better query complexity.

We also consider the notion of ``randomized reduction between PCPs'', defined below.
Bellare et al~\cite{BellareGS98} considered the notion of a randomized reduction $R$ between two promise problems given by sets $(A_1, B_1)$ and $(A_2, B_2)$. A randomized polynomial time reduction $R$ from promise problems $(A_1, B_1) \leq_R (A_2, B_2)$ with error probability $p$ satisfies:
\begin{enumerate}
    \item if $x \in A_1$ then w.p. $\geq 1-p$, $R(x) \in A_2$.
    \item if $x \in B_1$ then w.p. $\geq 1-p$, $R(x) \in B_2$.
\end{enumerate}

This notion naturally extends to PCP complexity classes. We give a randomized reduction between PCP classes with imperfect and perfect completeness.

\begin{theorem}\label{thm:rand-red}
Let $c, s \in (0,1), s < c$ be constants then there exists a constant $s' \in (0,1)$ depending only on $c, s$ such that,
$$\PCP_{c, s}[r, q] \leq_R \PCP_{1,s'}[r, q +O_{s,c}(\log r)]$$ with probability $1-2^{-\Omega(r)}$. Furthermore if the original proof size was $n$ then the final proof size is $n+O(2^r)$.
\end{theorem}
\end{paragraph}

\begin{paragraph}{\bf Gap-ETH without perfect completeness}
We study the relation between time complexities of approximating satisfiable instances of MAX 3SAT versus that of approximating unsatisfiable instances. We first show the equivalence of the Gap-ETH conjecture with perfect and imperfect completeness. We formally state the Gap-ETH conjecture below:
\begin{conjecture}[Gap Exponential-Time Hypothesis (Gap-ETH)~\cite{dinur-gap-eth,gap-eth2}] For some constants $\delta,\epsilon > 0$, no algorithm can, given a 3-SAT formula $\phi$ on $n$ variables and $m = O(n)$ clauses, solve the decision problem MAX 3-SAT$(1, 1-\epsilon)$ in $O(2^{\delta n})$ time.
\end{conjecture}
There are many versions of the Gap-ETH conjecture that one can consider. Many works study the randomized Gap-ETH conjecture which says that there are not even any randomized algorithms that can decide Max 3-SAT$(1,1-\epsilon)$. We show the following theorem:

\begin{theorem}\label{thm:equ}
If there exists a randomized (with no false positives) $2^{o(n)}$ time algorithm for MAX 3SAT$(1,1-\gamma)$ for all constant $\gamma > 0$ then there exists a randomized (with no false positives) $2^{o(n)}$ time algorithm for MAX 3SAT$(s(1+\epsilon),s)$ for all constants $s, \epsilon > 0$.
\end{theorem}

As the original Gap-ETH hypothesis~\cite{dinur-gap-eth,gap-eth2} talks about deterministic algorithms we would prefer to get a deterministic reduction between these two problems.

We can get more fine-grained results relating the time-complexities of approximating MAX 3SAT with perfect and imperfect completeness using Theorem~\ref{thm:rand-red} stated earlier.
\begin{corollary}
If there exists a $T(n)$ time algorithm for MAX 3SAT$(1,1-\delta)$ for all $\delta > 0$ then there exists a $T(n(\log\log n)^{O(1)})$ time randomized algorithm for MAX 3SAT$(1-\epsilon,1-\gamma)$ for all $\epsilon,\gamma$ such that, $0< \epsilon < \gamma$. 
\end{corollary}
\end{paragraph}

\subsection{Previous work}
Bellare et al~\cite{BellareGS98} also studied the problem of transforming probabilistically checkable proofs with imperfect completeness to those with perfect completeness. Their techniques do not yield any inclusions for PCP classes. They proved the following randomized reduction between PCP classes:

$$\PCP_{c,s}[r, q] \leq_R \PCP_{1,rs/c}[r, qr/c]$$

For constant $c$ and $r = \omega(1)$, they lose a multiplicative factor of $r$ in the soundness, which makes the theorem non-trivial only when $s = o(1)$.


\section{Preliminaries}
We will use the following notation:
\begin{paragraph}{\bf Notation:}
$\Thr{\delta}(x_{1},\ldots,x_{n})$ = threshold at $\delta$-fraction taken on the set of bits $\{x_{1},\ldots,x_{n}\}$.
We also use $\Thr{\delta}(x|_S)$ to mean that the threshold is with respect to the bits of $x$ restricted to $S \subseteq [n]$ and sometimes drop the $x$ and $\delta$ to use $\Thr{}(S)$, when the inputs and the fraction being used is clear from context. $exp(x)$ refers to $e^{x}$. For a string $x \in \{0, 1\}^n$, let $\bar{x} = \frac{1}{n}\sum_i x_i$, denote the average number of 1's in $x$.

MAX $k$-CSP$(c,s)$ is the promise problem of deciding whether there exists an assignment satisfying more than $c$-fraction given $k$-clauses or every assignment satisfies at most $s$ fraction of clauses. When the CSP is a 3SAT instance, it is denoted by MAX 3SAT$(c,s)$.
\end{paragraph}

We now discuss some standard probability results like the Chernoff bound and the Lov\'asz local lemma.
\begin{paragraph}{\bf Chernoff Bounds:}
\begin{enumerate}
    \item\label{c:2} Multiplicative Chernoff bound 1: Let $X = \frac{1}{n}X_i$, where $X_1,\ldots, X_n$ are independent random variables in $\{0,1\}$,  with $E[X] = \mu$. Then for all $\delta \geq 1$,
    $$\Pr\left[X > (1+\delta)\mu\right] \leq exp(-\Omega(\delta\mu))$$
    for $0 < \delta \leq 1$,
    $$\Pr[X > (1+\delta)\mu] \leq exp(-\delta^2\mu/3))$$
    $$\Pr[X < (1-\delta)\mu] \leq exp(-\delta^2\mu/2))$$
    \item\label{c:3} Multiplicative Chernoff bound 2: Let $X = \frac{1}{n}X_i$, where $X_1,\ldots, X_n$ are random variables in $\{0,1\}$,  with $E[X] = \mu$. Then for all $\delta \geq 2$,
    $$\Pr[X > (1+\delta)\mu] \leq \left(\frac{e^{\delta}}{(1+\delta)^{1+\delta}}\right)^\mu \leq \exp(-\Omega(\delta(\log(1/\delta))\mu)).$$
\end{enumerate}
\end{paragraph}

\begin{lemma}[Lov\'asz local lemma]\label{lem:lll}
Let $E_1, E_2, \ldots, E_n$ with $\Pr[E_i] = p$ be events such that any $E_i$ is independent of all but $d$ other events. Then if $pe(d+1) \leq 1$ then
$$\Pr\left[\bigcap_i \neg E_i\right] \geq (1-1/d)^n$$
\end{lemma}

Let us now define probabilistic proof systems. Firstly, we define the notion of an $(r, q)$-restricted verifier:
For integer valued functions $r(\cdot)$ and $q(\cdot)$, a verifier is said to be $(r, q)$-restricted if
on every input of length $n$, it tosses at most $r(n)$ coins and queries the proof in at most $q(n)$ bits non-adaptively.

\begin{definition}[PCP]
For integer-valued functions $r(\cdot), q(\cdot)$ and functions $c(\cdot),s(\cdot)$ mapping to $[0,1]$, the class $\PCP_{c,s}[r,q]$ consists of all languages for which there exists an $(r,q)$-restricted non-adaptive verifier $V$ with the following properties:
\begin{enumerate}
    \item Completeness: For all $x \in L$, there exists a proof $\pi$ such that $V^{\pi}(x)$ accepts with probability at least $c$ (over the coin tosses of $V$).
    \item Soundness: For all $x \notin L$, for all proofs $\pi$, $V^{\pi}(x)$ accepts with probability at most $s$.
\end{enumerate}
\end{definition}

We now go to the notion of averaging samplers. Averaging samplers are used to derandomize the process of random sampling to estimate the average number of ones in a string $x \in \{0,1\}^n$, see survey of~\cite{goldreich-survey}. We use the expander sampler from~\cite{exp-samp, GoldreichW97} (Lemma 6.6 in the latter) and also described in~\cite{goldreich-survey}:

\begin{lemma}\label{lem:exp-samp}
The expander sampler with parameters $(\delta, \epsilon, N)$ is  an expander graph on $N$ vertices, such that the neighbors of a vertex $i$, specify a sample $S_i \subseteq [N]$. The set family $\mathcal{ES}(\delta,\epsilon,N) = \{S_i\}_{i=1}^{N}$ satisfies the following properties: 
\begin{enumerate}
    \item For all $i$, $\abs{S_i} = O\left(\frac{1}{\delta\epsilon^2}\right)$.
    \item For every $S_i$ the number of sets $S_j$ which intersect with it are $O\left(\frac{1}{\delta^2\epsilon^4}\right)$.
    \item For any string $x \in \{0, 1\}^N$, $\Pr\limits_{S \sim \mathcal{ES}(\delta,\epsilon,N)}\left[|\overline{(x|_{S})} - \overline{x}| > \epsilon\right] \leq \delta$, where $\overline{(x|_{S})}$ denotes the average of $x$ taken over the positions specified by $S$.
\end{enumerate}
\end{lemma}

We analyze the expander sampler given above and prove that one can get a sampler with the following properties. In the appendix, we provide a detailed proof.

\begin{theorem}[Sampler]\label{samp}
For all constants $\epsilon,\delta,\gamma$, there exists a constant $C$ such that, there is a set family $\mathcal{S}(\epsilon,\delta,\gamma, N)= (S_i)_{i=1}^{N/2}$ on $[N]$ with the following properties:
 \begin{enumerate}
    \item \label{exp-samp-sound2} For any string $x \in \{0, 1\}^N$, $\Pr\limits_{S \sim \mathcal{S}}\left[|\overline{(x|_{S})} - \overline{x}| > \epsilon\right] \leq \delta$. 
    \item \label{exp-samp-comp2} For all $\eta < (1-\gamma)/2$, for any string $x \in \{0, 1\}^N$, where $\overline{x} \geq 1 - \eta$, we get that, ${\Pr\limits_{S \sim \mathcal{S}}\left[\overline{(x|_{S})} < \gamma\right] \leq \eta/2}$.
    \item \label{exp-samp-size2} For all $i$, $\abs{S_i} = C = O_{\epsilon, \delta, \gamma}(1)$. 
    \item The number of sets in $\mathcal{S}$ is $N/2$.
\end{enumerate}
\end{theorem}

\section{PCPs without perfect completeness}\label{sec:logn}

In this section we prove that PCPs with imperfect completeness can be converted to ones with perfect completeness with a mild blow-up in queries.

\subsection{Reductions with minimal Query Blow-up}

We first show a reduction that preserves the 
randomness complexity while losing an additive factor in the queries.

\begin{reminder}{Theorem~\ref{thm:pcp}} 
For all constants $c, s \in (0,1), s < c$, there exists a constant $s' \in (0,1)$, such that for all integer-valued functions $r(\cdot),q(\cdot)$, the following is true:
$$\PCP_{c, s}[r, q] \subseteq \PCP_{1,s'}[r, q +O_{s,c}(r)].$$ 

Furthermore if the original proof size was $n$, then the final proof size will be $n+O(2^r)$.\\
\end{reminder}
For notational simplicity we will prove that:
$$\PCP_{9/10,6/10}[r, q] \subseteq \PCP_{1,9/10}[r, q +O(r)],$$ with proof size $n+O(2^r)$. All constants that follow are universal constants, although in full generality, they only depend on $c,s$ that we have fixed to $(9/10,6/10)$.

The rest of this section is devoted to the proof of this theorem. The main idea here is to build a ``robust circuit'' of small depth, using threshold gates of small fan-in, over the proof oracle of the original PCP. We then ask the new prover to provide the original proof and along with that, also ask for what each gate in the circuit evaluates to, when provided the original clause evaluations as input. As discussed earlier, the circuit boosts the fraction of ones in every layer, for inputs $x$ that satisfy $\overline{x} \geq 9/10$, while maintaining the fraction of ones for inputs that satisfy $\overline{x} \leq 7/10$. We need to do this boosting step by step so that the fan-in does not blow up, and also need to use threshold gates that take ``random'' subsets of inputs from the previous layer, so that the ones in the input get distributed across all the gates. We get rid of the random subsets, by using any standard sampler over the gates of the previous layer.

Let us now describe the circuit more formally. Later we will give a way to get complete PCPs from incomplete ones using this circuit. 

\begin{mdframed}
Description of Circuit $\Gamma_m(\cdot)$:
\begin{itemize}
    \item The circuit has $d = \log{m}$ layers, $L_1,\ldots, L_d$, with layer $i$ composed of $w_i = m/2^{i}$ gates denoted by $L_{i1},\ldots,L_{i w_i}$. The zeroth layer $L_0$ is the $m$ inputs to the circuit.
    \item Every gate $L_{(i+1)j}$ is a threshold gate $\Thr{0.8}$. Let the set family given by the sampler from Theorem~\ref{samp} on $w_{i}$ nodes with parameters\\ $\mathcal{S}(1/10, 6/10, 8/10, w_{i}) = (S_{(i+1)j})_{j = 1}^{w_{i+1}}$. Let $L_{(i+1)j} = \Thr{0.8}(L_{i}|_{S_{(i+1)j}})$. By property~\ref{exp-samp-size2} of expander sampler fan-in $=\abs{S_{(i+1)j}} = O(1)$.
\end{itemize}
\end{mdframed}

We now use this circuit to give our main reduction.

\begin{proof}[Proof of Theorem~\ref{thm:pcp}]
Let $L \subseteq \Bool{*}$ be a language in $\PCP_{9/10,6/10}[r, q]$ via the proof system $\mathcal{P} = (\Pi, Q)$, where $\Pi$ and $Q$ denote the proof and the set of queries. We can now use the equivalence between MAX $q$-CSP$(c,s)$ and PCPs, to get a set of clauses $\mathcal{C} = \{C_1,\ldots, C_m\}$ of width $q$, for $m = 2^r$, such that $L \leq$ MAX $q$-$\mathcal{C}(9/10,6/10)$. (When $y \in L$, then there exists an assignment $x$, such that $9/10$-fraction of the clauses when evaluated on $x$ output $1$, whereas when $y \notin L$, for every assignment $x$, at most $6/10$ of the clauses evaluate to $1$.)

To prove the theorem, we will give a new proof system $\mathcal{P}' = (\Pi', Q')$ for $L$, that has perfect completeness and soundness equal to $9/10$. We will transform $\mathcal{P}$ using the circuit $\Gamma_m(\cdot)$ described above, to get $\mathcal{P}'$. We consider the circuit $\Gamma_m(C_1(\Pi), \ldots, C_m(\Pi))$ and ask the new prover to give one bit for every gate of the circuit. More precisely, we ask the new prover to give bits of $\Pi$ (interpreted as an assignment $x \in \Bool{n}$ for the MAX $q$-CSP: $\mathcal{C}$) and in addition gives bits for every layer in the circuit $\Gamma_m$:
$$\ell_i = \{\ell_{i1},\ldots, \ell_{iw_i}\}, \forall i \in \{0,1,\ldots,d\}.$$ 

These bits are supposed to correspond to a correct evaluation of the circuit $\Gamma_m$ when given $(C_1(x), \ldots, C_m(x))$ ($\Pi = x$) as input. That is, ideally the prover should give us, $\ell_{0j} = C_j(x), \forall j \in [m]$ and $\ell_{(i+1)j} = L_{(i+1)j}(\ell_{i}), \forall i \in [d],j \in w_i$, where $L_{(i+1)j}(\ell_{i})$ denotes the gate $L_{(i+1)j}$ evaluated on the output bit vector $\ell_i$ of the previous layer. We probabilistically test this using a new set of queries $Q'$, described below.

\paragraph{\bf Verifier Checks ($Q'$):}
For notational simplicity in describing the queries of the new verifier, we will do the following. For each layer $i$ (that has $m/2^i$ gates), consider $2^{i}$ copies of the set of gates $L_i$, and let this new set be denoted by $L_{i1}',\ldots,L_{im}'$ with corresponding proof bits by $\ell_i' = \{\ell_{i1}',\ldots,\ell_{im}'\}$ and each gate having its set of inputs $(S'_{i1}, \ldots, S'_{im})$. Note that this duplication of bits/gates is only for description of the queries, and the prover will only give $m/2^{i}$ bits for every layer $i$. 

Intuitively, we will check whether every gate is correct with respect to its immediate inputs (from the layer below it) and whether the final gate (on the topmost layer) evaluates to $1$. To do so, the verifier tosses $\log m$ random coins and on random string $j \in [m]$, it checks whether the following is true:
$$Q'_j := \left(C_j(x) \stackrel{?}{=} \ell'_{0j}\right) \land \left(L'_{1j}\left(\ell_{0}\right) \stackrel{?}{=} \ell'_{1j}\right) \ldots \land \left(L'_{d j}\left(\ell_{d-1}\right) \stackrel{?}{=} \ell'_{d j}\right) \land \ell'_{d j}$$
where the clause $(L'_{ij}(\ell_{i-1}) \stackrel{?}{=} \ell'_{ij})$ outputs 1 iff $(L'_{ij}(\ell_{i-1})$ equals $\ell'_{ij})$. 
As explained earlier, each of the clauses, checks whether the gate $L'_{ij}$ is correct, with respect to its input layer $\ell_{(i-1)}$. Notice here that each check $Q_j$, checks one gate in every layer and furthermore these checks are uniform across a layer, i.e. every gate in a layer is checked with the same probability.

To perform the check above, we query the proof bits  $\ell_{i-1}|_{S'_{ij}}$, making a constant number of queries, since the fanin of every gate is a fixed constant, i.e. $L_{ij}$ has fanin  $|S'_{ij}| = O(1)$. We then evaluate the threshold gate $L_{ij}$ on these bits and take the $\wedge$ across the layers. The check $(C_j(x) \stackrel{?}{=} \ell'_{0j})$ needs to query $q$ queries to $x$, hence the total number of queried proof bits is $q + O(\log m) = q + O(r)$. Further note that the randomness complexity of the verifier remains the same as before  i.e. $=r = \log m$.

We now prove the completeness and soundness of the protocol $\mathcal{P}'$.

\paragraph{Completeness:} If the original proof system $\mathcal{P}$ had completeness $9/10$, then there exists a proof $\Pi = x$ which satisfies $9/10$ of the clauses $\mathcal{C}$. The new prover can give us the bit vectors, $x$ and in addition the evaluations of the circuit $\Gamma(C_1(x), \ldots, C_m(x))$, i.e. $x, \ell_0 := (C_j(x))_{j=1}^m$ and $\ell_i := (L_{ij}(\ell_{i-1}))_{j = 1}^m$.

In Lemma~\ref{lem:comp-log}, we prove that, $\overline{\ell_i} \geq 1 - \frac{2^{-i}}{10}$. Since $d = \log m$ and the number of gates on level $d$ is $O(1)$, we get that the fraction of $1$s in $z_d$ is $\geq 1 - 1/m$, which gets rounded to $1$, since there is only one gate in the topmost layer. Since every query $Q'_j$ checks the consistency of a set of gates and if the bit $\ell_{dj} = 1$, we get completeness equals 1.

\paragraph{Soundness:} If the original proof system $\mathcal{P}$ had soundness $6/10$, then for all proofs $\Pi$ that the prover might give, $\Pi$ satisfies $\leq 6/10$ of the clauses $\mathcal{C}$. Let $\Pi' = (x, \ell_0,\ldots,\ell_{d})$ be the proof provided by the new prover. 

Let $z_0 := (C_j(x))_{j=1}^m$ and $z_{i+1} := (L_{(i+1)j}(\ell_{i}))_{i=1}^{w_i}$ be the true local evaluations. Note here that, $z_{i+1}$ is the evaluation bits of layer $L_{i+1}$ evaluated on the bits that the prover provides in the previous layer, $\ell_{i}$. By the soundness of $\mathcal{P}$ we get that $x$ satisfies at most $6/10$ of $\mathcal{C}$ which means that $\overline{z_0} \leq 6/10$. 

Now we have two cases: 
\begin{enumerate}
    \item The prover provided the bit vectors $\ell_i$ such that they agree with the true evaluations $z_i$ in most places, i.e.
    $$\forall i, \Pr\limits_{j \sim [w_i]}\left[\ell_{ij} \neq z_{ij}\right] \leq 1/10$$ 
    Hence we have that $\overline{\ell_0} \leq \overline{z_0} + 1/10 \leq 7/10$. Lemma~\ref{lem:soundness-log} gives us that for $\overline{\ell_i} \leq 7/10 $,  $\overline{L_{i+1}(\ell_i)} \leq 6/10$ and therefore $\overline{z_{i+1}} \leq 6/10$. Hence we get that by induction, for all $i$, $\overline{z_i} \leq 6/10$ and $\overline{\ell_i} \leq 7/10$, and more importantly $\overline{\ell_d} \leq 7/10$. Recall that our verifier checks are uniform over the every layer, and since $\ell_{dj} = 1$ is required for verifier's $j^{th}$ check, $Q_j$ to succeed, we get that soundness is $\leq 7/10$.
    
    \item There exists a layer $i \in \{0,\ldots,d\}$ such that:
    $$\Pr\limits_{j \sim [w_i]}\left[\ell_{ij} \neq z_{ij}\right] > 1/10.$$
    Since $z_{ij}$'s are the correct evaluations, the above implies that, the prover's proof will fail the local checks in $1/10$-fraction of the gates of layer $i$. 
    Since the verifier checks are uniform over the gates of every layer, (i.e. they check the gate of each layer with the same probability), the verifier checks the incorrect gates with probability at least $1/10$. Hence the soundness in this case is $\leq 9/10$.
\end{enumerate}  
Note that one of these cases has to occur, hence the overall soundness is the maximum of the two cases, i.e. $\leq 9/10$.

\paragraph{\bf Proof Length:} Every layer $L_i$ has width $m/2^i$. Thus the total number of gates in the circuit is $m+m/2 + \ldots = O(m) = O(2^r)$. Since $\Pi'$ consists of the original proof appended with the circuit evaluations, the proof length is $n+O(2^r)$.
\end{proof}

We now complete the proofs of completeness and soundness in Theorem~\ref{thm:pcp}.

\begin{lemma}[Completeness]\label{lem:comp-log}
Let $y_0 \in \Bool{m}$ be such that $\overline{y_0} \geq 9/10$. Let $y_i \in \Bool{w_i}$ denote the output string of layer $i$, when $\mathcal{C}$ is evaluated with the zeroth layer set to $y_0$. Then we have that for all $i$,  $y_i$ satisfies $\overline{y_i} \geq 1 - \frac{2^{-i}}{10}$.
\end{lemma}

\begin{proof}
We will prove the lemma by induction on $i$. Note that the base case $i=0$, holds trivially. Now consider the $(i+1)^{th}$ layer of the circuit and the gates $L_{(i+1)j}$ that take as input the set $S_{(i+1)j}$ corresponding to the expander sampler on $w_i$ bits. By the induction hypothesis we have that $y_i$ is such that $\overline{y_i} \geq 1 - \frac{2^{-i}}{10}$. By the expander sampler property~\ref{exp-samp-comp2} with parameters $(1/10, 6/10, 8/10, w_{i})$ we get that, 
$$\Pr\limits_{j \sim [w_{i+1}]}\left[L_{(i+1)j}(y_i) = \Thr{}(y_i|_{S_{(i+1)j}}) = 0\right] \leq \Pr\limits_{j \sim [w_{i+1}]}\left[\overline{(y_{i}|_{S_{(i+1)j}})} < 0.8\right] \leq \left(\frac{1}{2}\right)\left(\frac{2^{-i}}{10}\right).$$

Which directly implies
$$\Pr\limits_{j \sim [w_{i+1}]}\left[L_{(i+1)j}(y_i) = 1\right] \geq 1 - \frac{2^{-i-1}}{10} \Leftrightarrow \overline{y_{i+1}} \geq 1 - \frac{2^{-i-1}}{10}$$
which completes the induction.
\end{proof}

\begin{lemma}[Soundness]\label{lem:soundness-log} 
Let $y_{i} \in \Bool{w_i}$ denote an instantiation of the output gates of layer $i$ with $\overline{y_i} \leq 7/10$. Let $y_{i+1} = L_{i+1}(y_i)$ denote the output of layer $i+1$ when evaluated on the string $y_i$. Then we have that $y_{i+1}$ satisfies $\overline{y_{i+1}} \leq 6/10$.
\end{lemma}
\begin{proof}
Recall that in the circuit, the gate $L_{(i+1)j}$ took as input the set $S_{(i+1)j}$ corresponding to the sampler on $w_i$ bits. By the expander sampler property~\ref{exp-samp-sound2}, with parameters $(1/10, 6/10, 8/10, w_{i})$ we get that, for any string $y_i \in \Bool{w_i}$ with $\overline{y_i} \leq 7/10$:

$$\Pr\limits_{j \sim [w_{i+1}]}\left[|\overline{(y_i|_{S_{(i+1)j}})} - 7/10| > 1/10\right] \leq \Pr\limits_{j \sim [w_{i+1}]}\left[L_{(i+1)j}(y_i) = \Thr{}(y_i|_{S_{(i+1)j}}) = 1\right] \leq 6/10$$

which directly implies $\overline{y_{i+1}} \leq 6/10$ completing the proof.

\end{proof}

Theorem~\ref{thm:pcp} implies the following transformation from linear sized PCPs with imperfect completeness to linear sized PCPs with perfect completeness.
\begin{corollary}
If $NTIME[O(n)] \subseteq \PCP_{9/10,6/10}[\log n + O(1), q]$ then $NTIME[O(n)] \subseteq \PCP_{1,9/10}[\log n + O(1), q + O(\log n)]$. 
\end{corollary}
\section{Randomized reductions between PCPs}\label{sec:loglog}

In this section we prove that PCPs with imperfect completeness can be reduced using randomness to ones with perfect completeness with a lesser blow-up in queries compared to Section~\ref{sec:logn}. We construct a  circuit similar to the one in the previous section, but this time we use a randomized circuit to get better parameters and show that our reduction works with high probability. 
\subsection{Randomized Reductions with minimal Query Blow-up}

\begin{reminder}{Theorem~\ref{thm:rand-red}} 
For all constants $c, s \in (0,1), s < c$, there exists a constant $s' \in (0,1)$, such that for all integer-valued functions $r(\cdot),q(\cdot)$, the following is true:
$$\PCP_{c, s}[r, q] \leq_R \PCP_{1,s'}[r, q +O_{s,c}(\log r)].$$ 

Furthermore if the original proof size was $n$, then the final proof size will be $n+O(2^r)$.\\
\end{reminder}
For notational simplicity we will prove that:
$$\PCP_{9/10,6/10}[r, q] \leq_R \PCP_{1,9/10}[r, q +O(\log r)],$$ with proof size $n+O(2^r)$. All constants that follow are universal constants, although in full generality, they only depend on $c,s$ that we have fixed to $(9/10,6/10)$.

This immediately implies the following corollary using the query reduction\footnote{This result reduces queries to a constant but blows-up the proof size.} result by Dinur~\cite{dinur},
\begin{corollary}
If there exists a $T(n)$ time algorithm for MAX 3SAT$(1,1-\delta)$ for all $\delta > 0$ then there exists a $T(n(\log\log n)^{O(1)})$ time randomized algorithm for MAX 3SAT$(1-\epsilon,1-\gamma)$ for all $\epsilon,\gamma, 0< \epsilon < \gamma$. 
\end{corollary}

The rest of this section is devoted to the proof of theorem~\ref{thm:rand-red}. The main idea as in Theorem~\ref{thm:pcp} is to build a ``robust circuit'' of small depth, using threshold gates of small fan-in, over the proof oracle of the original PCP. We then ask the new prover to provide the original proof and along with that, also ask for what each gate in the circuit evaluates to, when provided the original clause evaluations as input. As discussed earlier, the circuit boosts the fraction of ones in every layer, for inputs $x$ that satisfy $\overline{x} \geq 9/10$, while maintaining the fraction of ones for inputs that satisfy $\overline{x} \leq 7/10$. We need to do this boosting step by step so that the fan-in does not blow up, and also need to use threshold gates that take random subsets of inputs from the previous layer, so that the ones in the input get distributed across all the gates.

Let us now describe the circuit more formally. Later we will give a way to get complete PCPs from incomplete ones using this circuit. 

\begin{mdframed}
Description of Circuit $\Gamma_m(\cdot)$:
\begin{itemize}
    \item The circuit has $d = \log\log{m}$ layers, $L_1,\ldots, L_d$, with layer $i$ composed of $w_i = m/2^{i}$ gates denoted by $L_{i1},\ldots,L_{i w_i}$. The zeroth layer $L_0$ is the $m$ inputs to the circuit.
    \item Every gate $L_{ij}$ is a threshold gate $\Thr{0.8}$. A gate $L_{ij}$ takes as inputs a random set of $f = O(1)$ gates from the previous layer $L_{i-1}$, i.e. we pick a uniformly random set $S_{ij}$ of size $f$, (sampled with replacement) from $[m/2^{i-1}]$ and connect gate $L_{ij}$ with gates $L_{(i-1)k}, \forall k \in S_{ij}$.
\end{itemize}
\end{mdframed}

We now use this circuit to give our main reduction.

\begin{proof}[Proof of Theorem~\ref{thm:rand-red}]
Let $L \subseteq \Bool{*}$ be a language in $\PCP_{9/10,6/10}[r, q]$ via the proof system $\mathcal{P} = (\Pi, Q)$, where $\Pi$ and $Q$ denote the proof and the set of queries. We can now use the equivalence between MAX $q$-CSP$(c,s)$ and PCPs to get a set of clauses $\mathcal{C} = \{C_1,\ldots, C_m\}$ of width $q$, for $m = 2^r$, such that $L \leq$ MAX $q$-$\mathcal{C}(9/10,6/10)$. (When $y \in L$, then there exists an assignment $x$, such that $9/10$-fraction of the clauses when evaluated on $x$ output $1$, whereas when $y \notin L$, for every assignment $x$, at most $6/10$ of the clauses evaluate to $1$.)

To prove the theorem, we will give a new proof system $\mathcal{P}' = (\Pi', Q')$ for $L$, that has perfect completeness and soundness equal to $9/10$. We will transform $\mathcal{P}$ using the circuit $\Gamma_m(\cdot)$ described above, to get $\mathcal{P}'$. We consider the circuit $\Gamma_m(C_1(\Pi), \ldots, C_m(\Pi))$ and ask the new prover to give one bit for every gate of the circuit. More precisely, we ask the new prover to give bits of $\Pi$ (interpreted as an assignment $x \in \Bool{n}$ for the MAX $q$-CSP: $\mathcal{C}$) and in addition gives bits for every layer in the circuit $\Gamma_m$:
$$\ell_i = \{\ell_{i1},\ldots, \ell_{iw_i}\}, \forall i \in \{0,1,\ldots,d\}.$$ 

These bits are supposed to correspond to a correct evaluation of the circuit $\Gamma_m$ when given $(C_1(x), \ldots, C_m(x))$ ($\Pi = x$) as input. That is, ideally the prover should give us, $\ell_{0j} = C_j(x), \forall j \in [m]$ and $\ell_{(i+1)j} = L_{(i+1)j}(\ell_{i}), \forall i \in [d],j \in w_i$, where $L_{(i+1)j}(\ell_{i})$ denotes the gate $L_{(i+1)j}$ evaluated on the output bit vector $\ell_i$ of the previous layer. We probabilistically test this using a new set of queries $Q'$, described below.

\paragraph{\bf Verifier Checks ($Q'$):}
For notational simplicity in describing the queries of the new verifier, we will do the following. For each layer $i$ (that has $m/2^i$ gates), consider $2^{i}$ copies of the set of gates $L_i$, and let this new set be denoted by $L_{i1}',\ldots,L_{im}'$ with corresponding proof bits by $\ell_i' = \{\ell_{i1}',\ldots,\ell_{im}'\}$ and each gate having its set of inputs $(S'_{i1}, \ldots, S'_{im})$. Note that this duplication of bits/gates is only for description of the queries, and the prover will only give $m/2^{i}$ bits for every layer $i$. 

Intuitively, we will check whether every gate is correct with respect to its immediate inputs (from the layer below it) and whether the final gate (on the topmost layer) evaluates to $1$. To do so, the verifier tosses $\log m$ random coins and on random string $j \in [m]$, it checks whether the following is true:
$$Q'_j := (C_j(x) \stackrel{?}{=} \ell'_{0j}) \land (L'_{1j}(\ell_{0}) \stackrel{?}{=} \ell'_{1j}) \ldots \land (L'_{d j}(\ell_{d-1}) \stackrel{?}{=} \ell'_{d j}) \land \ell'_{d j},$$
where the clause $(L'_{ij}(\ell_{i-1}) \stackrel{?}{=} \ell'_{ij})$ outputs 1 iff $(L'_{ij}(\ell_{i-1})$ equals $\ell'_{ij})$. 
As explained earlier, each of the clauses, checks whether the gate $L'_{ij}$ is correct, with respect to its input layer $\ell_{(i-1)}$. Notice here that each check $Q_j$, checks one gate in every layer and furthermore these checks are uniform across a layer, i.e. every gate in a layer is checked with the same probability.

To perform the check above, we query the proof bits  $\ell_{i-1}|_{S'_{ij}}$, making a constant number of queries, since the fanin of every gate is a fixed constant, i.e. $L_{ij}$ has fanin  $|S'_{ij}| = O(1)$. We then evaluate the threshold gate $L_{ij}$ on these bits and take the $\wedge$ across the layers. The check $(C_j(x) \stackrel{?}{=} \ell'_{0j})$ needs to query $q$ queries to $x$, hence the total number of queried proof bits is $q + O(\log\log m) = q + O(\log r)$. Further note that the randomness complexity of the verifier remains the same as before, $=r = \log m$.

We now prove the completeness and soundness of the protocol $\mathcal{P}'$. Since the reduction is randomized, this boils down to proving that, 1) Completeness: given a Max $q$-CSP that was $c$-satisfiable, with high probability it gets mapped to a Max $q'$-CSP that is perfectly satisfiable and 2) Soundness: given a Max $q$-CSP that was at most $s$-satisfiable, with high probability it gets mapped to a Max $q'$-CSP that is at most $s'$-satisfiable.

\paragraph{Completeness:} If the original proof system $\mathcal{P}$ had completeness $9/10$, then there exists a proof $\Pi = x$ which satisfies $9/10$ of the clauses $\mathcal{C}$. The new prover can give us the bit vectors, $x$ and in addition the evaluations of the circuit $\Gamma(x)$, i.e. $x, \ell_1 := (C_j(x))_{j=1}^m$ and $\ell_i := (L_{ij}(\ell_{i-1}))_{j = 1}^m$.
In Lemma~\ref{lem:comp-loglog}, we prove that with probability $\geq 1-1/m^{1/4}$, $\overline{\ell_d} = 1$. Since every query $Q'_j$ checks the consistency of a set of gates and if the bit $\ell_{dj} = 1$ we get that with probability $1-1/m^{1/4} = 1-2^{-\Omega(r)}$, completeness equals 1.

\paragraph{Soundness:} 

We will call a circuit $\Gamma_m(\mathcal{C})$ ``good'' if the following property holds:

For all layers $i$, $\forall \ell_i \in \{0,1\}^{w_i}$ such that $\overline{\ell_i} \leq 7/10$, the circuit is such that $\overline{L_{i+1}(\ell_i)} \leq 6/10$. (Recall that $L_{i+1}(z)$ denotes the output of layer $L_{i+1}$ when evaluated on the string $z$.)

Lemma~\ref{lem:comp-loglog} gives us that,
$$\Pr[\forall \ell_i \text{ with } \overline{\ell_i} \leq 7/10, \overline{L_{i+1}(\ell_i)} \leq 6/10] \geq 1 - 2^{-m/2^{i}}$$

Taking a union bound over the layers of the circuit, we get that,
\begin{align*}
\Pr[\Gamma_m(\mathcal{C}) \text{ is good }] &= \Pr[\forall i, \forall \ell_i \text{ with } \overline{\ell_i} \leq 7/10, \overline{L_{i+1}(\ell_i)} \leq 6/10]\\
&\geq 1 - (\log\log m)2^{-m/2^{d}} \\
&\geq 1-2^{-\sqrt{m}} = 1-2^{-\Omega(r)}
\end{align*}
We will now show that if the randomized circuit $\Gamma_m(\mathcal{C})$ is good then the new PCP is sound. Since the circuit is good with high probability, showing this is enough to complete the randomized reduction claimed in Theorem~\ref{thm:rand-red}.

From now on, we will assume that the circuit is good. If the original proof system $\mathcal{P}$ had soundness $6/10$, then for all proofs $\Pi$ that the prover might give, $\Pi$ satisfies $\leq 6/10$ of the clauses $\mathcal{C}$. Let $\Pi' = (x, \ell_0,\ldots,\ell_{d})$ be the proof provided by the new prover. 

Let $z_0 := (C_j(x))_{j=1}^m$ and $z_{i+1} := (L_{(i+1)j}(\ell_{i}))_{i=1}^{w_i}$ be the true local evaluations. Note here that, $z_{i+1}$ is the evaluation bits of layer $L_{i+1}$ evaluated on the bits that the prover provides in the previous layer, $\ell_{i}$. By the soundness of $\mathcal{P}$ we get that $x$ satisfies at most $6/10$ of $\mathcal{C}$ which means that $\overline{z_0} \leq 6/10$. 

Now we have two cases: 
\begin{enumerate}
    \item The prover provided the bit vectors $\ell_i$ such that they agree with the true evaluations $z_i$ in most places, i.e.
    $$\forall i, \Pr\limits_{j \sim [w_i]}[\ell_{ij} \neq z_{ij}] \leq 1/10.$$ 
    Hence we have that $\overline{\ell_0} \leq \overline{z_0} + 1/10 \leq 7/10$. Lemma~\ref{lem:soundness-loglog} gives us that for $\overline{\ell_i} \leq 7/10 $,  $\overline{L_{i+1}(\ell_i)} \leq 6/10$ and therefore $\overline{z_{i+1}} \leq 6/10$. Hence we get that by induction, for all $i$, $\overline{z_i} \leq 6/10$ and $\overline{\ell_i} \leq 7/10$, and more importantly $\overline{\ell_d} \leq 7/10$. Recall that our verifier checks are uniform over the every layer, and since $\ell_{dj} = 1$ is required for verifier's $j^{th}$ check, $Q_j$ to succeed, we get that soundness is $\leq 7/10$.
    
    \item There exists a layer $i \in \{0,\ldots,d\}$ such that:
    $$\Pr\limits_{j \sim [w_i]}[\ell_{ij} \neq z_{ij}] > 1/10.$$
    Since $z_{ij}$'s are the correct evaluations, the above implies that, the prover's proof will fail the local checks in $1/10$-fraction of the gates of layer $i$. 
    Since the verifier checks are uniform over the gates of every layer, (i.e. they check the gate of each layer with the same probability), the verifier checks the incorrect gates with probability at least $1/10$. Hence the soundness in this case is $\leq 9/10$.
\end{enumerate}  
Note that one of these cases has to occur, hence the overall soundness is the maximum of the two cases, i.e. $\leq 9/10$.

\paragraph{\bf Proof Length:} Every layer $L_i$ has width $m/2^i$. Thus the total number of gates in the circuit is $m+m/2 + \ldots = O(m) = O(2^r)$. Since $\Pi'$ consists of the original proof appended with the circuit evaluations, the proof length is $n+O(2^r)$.
\end{proof}

We now complete the proofs of completeness and soundness claims used in the proof of Theorem~\ref{thm:rand-red}.

\begin{lemma}[Completeness]\label{lem:comp-loglog}
Let $y_0 \in \Bool{m}$ be such that $\overline{y_0} \geq 9/10$. Let $y_i \in \Bool{w_i}$ denote the output string of layer $i$, when $\mathcal{C}$ is evaluated on $y_0$. Then we have that with probability $\geq 1-1/m^{1/4}$ for all $i$,  $y_i$ satisfies $\overline{y_i} \geq 1 - \left(\frac{1}{10}\right)^{2^{i}}$ and hence $\overline{y_d} = 1$.
\end{lemma}

Notice here that the completeness $1 - \eta$ increases to $1 - (\eta)^2$ at each step, instead of $1 - \eta$ to $1 - \eta/2$, like it did in the previous section. This increase allows us to use only $\log\log m$ layers to get perfect completeness, albeit with high probability. Now we prove the lemma.

\begin{proof}
The statement is implied by proving that with probability $\geq 1-1/m^{1/4}$ for all $i$, $(1-\overline{y_{i+1}}) \leq (1-\overline{y_i})^2$.\\
We will prove the lemma by induction on $i$. Note that the base case $i=0$, holds trivially. Now consider the $(i+1)^{th}$ layer of the circuit and the gates $L_{(i+1)j}$ that take as input the set $S_{(i+1)j}$ corresponding to random sets of size $f$ from $[w_i]$.

By induction $\overline{y_i} \geq 1 - \left(\frac{1}{10}\right)^{2^{i}} \geq .9$ and $.2/(1-\overline{y_i}) \geq 2$. For a fixed $L_{(i+1)j}$, by the Chernoff bound~\ref{c:3} on number of $0$'s we get,
\begin{align*}
\Pr[L_{(i+1)j}(y_i) = \Thr{.8}(y_i|_{S_{(i+1)j}}) = 0] &= \Pr[\Thr{.2}((1-y_i)|_{S_{(i+1)j}}) = 1]\\  &\leq \exp\left(-\Omega\left(\left(.2/\left(1-\overline{y_i}\right)\right) \cdot \log\left(.2/\left(1-\overline{y_i}\right)\right) \cdot \left(1-\overline{y_i}\right)f\right)\right)\\ &= \exp(-\Omega(\log(.5/(1-\overline{y_i})) \cdot f))\\ 
&= \exp(\Omega(\log(5(1-\overline{y_i})) \cdot f))\\ 
&=  O\left(\left(1-\overline{y_i}\right)^3\right)
\end{align*}
for some large enough constant $f$.

Chernoff bound~\ref{c:2} over all the gates in $L_{i+1}$ for the number of 0's gives gives that,
\begin{align}
\begin{split}\label{e:1}
\Pr[(1-\overline{y_{i+1}}) \geq \left(1-\overline{y_i}\right)^2] &< \exp\left(-\Omega\left(\left(\left(1-\overline{y_i}\right)^2/\left(1-\overline{y_i}\right)^3\right) \cdot \left(1-\overline{y_i}\right)^3 \cdot \left(m/2^{i}\right)\right)\right)\\ &= \exp\left(-\Omega\left(\left(1-\overline{y_i}\right)^2\left(m/2^{i}\right)\right)\right)
\end{split}
\end{align}
As we have $\log{\log{m}}$ layers in the last layer the number of gates is $m/2^{\log{\log{m}}} > m/\log{m}$, hence
$$\Pr[(1-\overline{y_{i+1}}) \geq \left(1-\overline{y_i}\right)^2] < \exp\left(-\Omega\left(\left(1-\overline{y_i}\right)^2\left(m/\log m\right)\right)\right)$$
A Markov bound over all the gates in $L_{i+1}$ for the number of 0's gives gives that,
\begin{align}\label{e:2}
\Pr[(1-\overline{y_{i+1}}) \geq (1-\overline{y_i})^2] \leq \frac{\text{E}[1-\overline{y_{i+1}}]}{(1-\overline{y_i})^2} \leq O((1-\overline{y_i})).
\end{align}
Equations~\ref{e:1} and~\ref{e:2} together these bounds imply
\begin{align*}
\Pr[(1-\overline{y_{i+1}}) \geq (1-\overline{y_i})^2] &\leq \min(\exp(-\Omega((1-\overline{y_i})^2(m/\log m))), O((1-\overline{y_i})))\\
&\leq \log^2(m)/\sqrt{m}
\end{align*}
Union bound over all $\log m$ layers gives probability $\leq (\log\log m)\log^2(m)/\sqrt{m}$ $\leq 1/m^{1/4}$.
Hence with probability $\geq 1-1/m^{1/4}$, $\overline{y_d} \geq 1 - \left(\frac{1}{10}\right)^{2^{\log\log(m)}} \geq 1-1/m^2$. As there are $\leq m$ gates at last layer this means with probability $\geq 1-1/m^{1/4}$, $\overline{y_d} = 1$.
\end{proof}

\begin{lemma}[Soundness]\label{lem:soundness-loglog} 
Let $y_{i} \in \Bool{w_i}$ denote an instantiation of the output gates of layer $i$ with $\overline{y_i} \leq 7/10$. Let $L_{i+1}(y_i)$ denote the output of layer $i+1$ when evaluated on the string $y_i$. Then with probability $1-2^{-m/2^{i}}$, for all $y_i$, $L_{i+1}(y_i)$ satisfies $\overline{L_{i+1}(y_i)} \leq 6/10$. Formally,

$$\Pr[\forall y_i \text{ with } \overline{y_i} \leq 7/10, \overline{L_{i+1}(y_i)} \leq 6/10] \geq 1 - 2^{-m/2^{i}}.$$
\end{lemma}

\begin{proof}
Fix a gate $L_{(i+1)j}$. Given that the fraction of $1$s in layer $i$ is at most $7/10$, using Chernoff bound~\ref{c:2}, we get that,

\begin{align}
\Pr[\Thr{0.8}(S_{(i+1)j}) = 1] &= \Pr[\frac{1}{f}\sum_{k \in S_{(i+1)j}}\ell_{ik} - 7/10 > 8/10=7/10(1+1/7)] \\
&< \exp(-(1/7)^2(7f/10)/3) \\
&< 1/f, 
\end{align}

for large enough constant $f$.

By applying Chernoff bound~\ref{c:3} (assuming large enough $f$) over all gates $L_{(i+1)j}$, we get that,
\begin{align*}
\Pr[\overline{L_{i+1}(y_i)} > 6/10] &< \exp\left(-\Omega\left(\left(6f/10\right)\left(\log\left(6f/10\right)\right)\left(m/\left(f2^{i+1}\right)\right)\right)\right)\\
&= \exp\left(-\Omega\left(\left(\log\left(6f/10\right)\right)\left(m/2^{i}\right)\right)\right)\\ &< \exp\left(-2m/2^{i}\right).
\end{align*}
 for large enough constant f.\\
Hence a union bound over all possible $2^{m/2^{i}}$ strings $y_i$ gives that,
$$\Pr[\exists y_i, \overline{L_{i+1}(y_i)} > 6/10] \leq 2^{m/2^{i}}e^{-2m/2^{i}} < 2^{-m/2^{i}}.$$
\end{proof}

\section{Gap-ETH without perfect completeness}

In this section we study the relation between the time complexities of approximating MAX 3-SAT with and without perfect completeness. We show that the Gap-ETH conjecture with and without perfect completeness is equivalent by giving an algorithm for approximating MAX 3-SAT without perfect completeness, that uses an algorithm for approximating MAX 3-SAT with perfect completeness as a subroutine and runs in $2^{o(n)}$-time iff the latter does so.

\subsection{Reduction for two-sided error randomized algorithms}
We first prove that Gap-ETH conjecture with and without perfect completeness are equivalent for randomized algorithms with two-sided error. We show this by showing that the Gap-ETH conjecture without perfect completeness is false if the one with perfect completeness is false.
\begin{theorem}\label{thm:equ2}
If there exists a randomized (two-sided error) $2^{o(n)}$ time algorithm for MAX 3SAT$(1,1-\gamma)$, for all constants $\gamma > 0$, then there exists a randomized (two-sided error) $2^{o(n)}$ time algorithm for MAX 3SAT$(s(1+\epsilon),s)$ for all constants $s, \epsilon$.
\end{theorem}

We will prove the above in its contrapositive form. Suppose there is a $2^{o(n)}$ algorithm for MAX 3SAT$(1,1-\gamma)$ for all constants $\gamma$. We will then show that for all constants $\epsilon,s,\delta$, there exists an algorithm for MAX 3SAT$(s(1+\epsilon),s)$ with running time less than $2^{\delta n}$. Our randomized algorithm for MAX 3SAT$(s(1+\epsilon), s)$ will use the algorithm for satisfiable MAX 3SAT$(1,1-\gamma)$ as a subroutine and run in time less than $2^{\delta n}$. The following lemma forms the crux of the proof.

\begin{lemma}\label{lem:gap-easy}
For all constant $s, \epsilon > 0$ there exists a large enough constant $k$, such that there exists a randomized reduction from MAX 3-SAT$(s(1+\epsilon),s)$ on $n$ variables and $O(n)$ clauses to MAX $O(k)$-CSP$(1,1/2)$ on $n$ variables and $O(n)$ clauses, such that:
\begin{itemize}
    \item If the original instance was a NO instance, then the reduction produces an instance which is not a NO instance with probability $\leq 2^{-n}$.
    \item If the original instance was a YES instance, then the reduction produces a YES instance with probability $\geq 2^{-n/k}$.
\end{itemize}
\end{lemma}

\begin{proof}
Let $\mathcal{C} = \{C_1,\ldots,C_m\}$ be a MAX 3SAT$(s(1+\epsilon),s)$ instance. We can assume without loss of generality, that $\epsilon < 1/100$, since the result for a smaller gap implies the result for a larger gap.

Let $(S_i)_{i = 1}^{n}$
be a set family in which every set $S_i$ is a random set of cardinality $k$
chosen by sampling with replacement from $[m]$. Consider new clauses $B_i$ such that each clause is a threshold gate: $B_i = \Thr{s(1+\epsilon/2)}(\mathcal{C}|_{S_i})$, where $\mathcal{C}$ denotes the vector $(C_1(x),\ldots,C_m(x))$.

Our final CSP will be over the original set of variables $x_i$. We will have a clause for each of the $n$ $B_i$'s. For the $i^{th}$ clause $B_i$, we will find the values of all $C_j$ such that $j \in S_i$ and then verify that their threshold value is $\geq s(1+\epsilon/2)$. Our query size is $3k$ as we find values for variables in $k$ clauses each of them on 3 variables.

\begin{paragraph}{\bf Soundness}
For a NO instance and a fixed assignment $x$ the fraction of clauses satisfied by $x$ is $\leq s$. By the Chernoff bound~\ref{c:2}, the probability that clause $B_i$ is satisfied is $\leq \exp\left(-\left(\eps/2\right)^2sk/3\right)$. The probability that at least half of the $B_i$'s are satisfied is at most, $\binom{n}{n/2}\exp\left(-\Omega\left(\eps^2skn\right)\right)$ which is less than $\exp(-2n)$, when $k$ is taken to be a large enough constant, depending only on $\epsilon, s$. Therefore by a union bound, the probability that there exists an assignment $x$ that satisfies at least half of the $B_i$'s is $\leq 2^n \exp(-2n) \leq 2^{-n}$.
\end{paragraph}

\begin{paragraph}{\bf Completeness}
For a YES instance there exists an assignment $x$ that satisfies $\geq s(1+\epsilon)$-fraction of the clauses. By the Chernoff bound~\ref{c:2} the probability that the clause $B_i$ is unsatisfied is $\leq \exp(-(\eps/3)^2sk/2)$ as $\epsilon < 1/100$. Therefore the probability that all the $B_i$'s are satisfied is $(1-\exp(-\Omega(\eps^2sk)))^{n} \geq (1-10/k)^{n}$ which is $\geq 2^{-n/k}$ when $k$ is a large enough constant.
\end{paragraph}
\end{proof}

\begin{proof}[Proof of Theorem~\ref{thm:equ2}]

The randomized algorithm for solving MAX 3SAT$(s(1+\epsilon), s)$ is as follows: We will run the reduction from Lemma~\ref{lem:gap-easy} $2^{n/k}n^2$ times and then convert the resulting MAX $O(k)$-CSP$(1,1/2)$ instances to MAX 3SAT$(1,1-\gamma)$ instances on $O(k2^kn)$ variables and $O(k2^kn)$ clauses where $\gamma$ is a constant depending on $k$. Then we run the $2^{o(n)}$ algorithm for MAX 3SAT$(1,1-\gamma)$  (still $2^{o(n)}$ as $k, \gamma$ are constants) on the resulting instances and if any of the outputs is YES we will also output YES.  

By repeating the algorithm for MAX 3SAT$(1,1-\gamma)$ $\poly(n)$ times we can assume the the probability that the algorithm errs is $\leq 2^{-n^2}$, hence we will assume this wlog.

\begin{align*}
\Pr[\text{Error on a YES instance}] &\leq \Pr[\text{Algorithm errs on one of the produced instances}]\\
&+\Pr[\text{None of the $2^{n/k}n^2$ runs produce a YES instance}]\\
&\leq 2^{-n^2}2^{n/k}n^2+(1-2^{-n/k})^{2^{n/k}n^2}\\
&\leq 2^{-n/2}
\end{align*}

\begin{align*}
\Pr[\text{Error on a NO instance}] &\leq \Pr[\text{Algorithm errs on one of the produced instances}]\\
&+\Pr[\text{On one of the $2^{n/k}n^2$ runs the output was not a NO instance}]\\
&\leq 2^{-n^2}2^{n/k}n^2+2^{n/k}n^22^{-n}\\
&\leq 2^{-n/2}
\end{align*}

Total running time $= 2^{n/k}n^22^{o(n)}$ which for large enough $k$ is $< 2^{\delta n}$. This gives us the desired contradiction.
\end{proof}

\subsection{Reduction for one-sided error randomized algorithms with no false positives}
We now prove that in fact Gap-ETH conjecture with and without perfect completeness are equivalent for randomized algorithms with no false positives.

\begin{reminder}{Theorem~\ref{thm:equ}} 
If there exists a randomized (with no false positives) $2^{o(n)}$ time algorithm for MAX 3SAT$(1,1-\gamma)$ for all constant $\gamma > 0$ then there exists a randomized(with no false positives) $2^{o(n)}$ time algorithm for MAX 3SAT$(s(1+\epsilon),s)$ for all constants $s, \epsilon > 0$.
\end{reminder}

As in the proof of Theorem~\ref{thm:equ2}, we will prove the above in its contrapositive form.
Suppose there is a $2^{o(n)}$ algorithm (with no false positives) for MAX 3SAT$(1,1-\gamma)$ for all constants $\gamma$. We will then show that for all constants $\epsilon,s,\delta$, there exists an algorithm (with no false positives) for MAX 3SAT$(s(1+\epsilon),s)$ with running time less than $2^{\delta n}$. Our randomized algorithm for MAX 3SAT$(s(1+\epsilon), s)$ will use the algorithm for satisfiable MAX 3SAT$(1,1-\gamma)$ as a subroutine and run in time less than $2^{\delta n}$. The following lemma which is a stronger version of Lemma~\ref{lem:gap-easy} with only one-sided error forms the crux of the proof.

\begin{lemma}\label{lem:gap}
For all constant $s, \epsilon > 0$ there exists a large enough constant $k$, such that there exists a randomized reduction from MAX 3SAT$(s(1+\epsilon),s)$ to MAX $O(k)$-CSP$(1,1/2)$ with $O(n)$ variables such that:
\begin{itemize}
    \item If the original instance was NO, then the reduction produces a NO instance.
    \item If the original instance was YES, then the reduction produces a YES instance with probability $\geq 2^{-n/k}$.
\end{itemize}
\end{lemma}

We will first prove Theorem~\ref{thm:equ} using the lemma given above. This proof is similar to the proof of Theorem~\ref{thm:equ2}.

\begin{proof}[Proof of Theorem~\ref{thm:equ}]
The randomized algorithm for solving MAX 3SAT$(s(1+\epsilon), s)$ is as follows: We will run the reduction from Lemma~\ref{lem:gap} $2^{n/k}n^2$ times and then convert the resulting MAX $O(k)$-CSP$(1,1/2)$ instances to a MAX 3SAT$(1,1-\gamma)$ instances on $O(k2^kn)$ variables and $O(k2^kn)$ clauses where $\gamma$ is a constant depending on $k$. Then we run the $2^{o(n)}$ algorithm for MAX 3SAT$(1,1-\gamma)$ algorithm (still $2^{o(n)}$ as $k, \gamma$ are constants) on them and if any of the outputs is YES we will also output YES. 

By repeating the algorithm $\poly(n)$ times we can assume the the probability that the algorithm errs (one sided  error) is $\leq 2^{-n^2}$, hence we will assume this wlog.

For a NO original instance we will always output a NO instance.
\begin{align*}
\Pr[\text{Error on a YES instance}] &\leq \Pr[\text{Algorithm errs on one of the produced instances}]\\
&+\Pr[\text{None of the $2^{n/k}n^2$ runs produce a YES instance}]\\
&\leq 2^{-n^2}2^{n/k}n^2+(1-2^{-n/k})^{2^{n/k}n^2)}\\
&\leq 2^{-n}
\end{align*}

Total running time $= 2^{n/k}n^22^{o(n)}$ which for large enough $k$ is $< 2^{\delta n}$. This gives us the desired contradiction.
\end{proof}

\begin{proof}[Proof of Lemma~\ref{lem:gap}]
Let $\mathcal{C} = \{C_1,\ldots,C_m\}$ be a MAX 3SAT$(s(1+\epsilon),s)$. We can assume without loss of generality, that $\epsilon < 1/100$, since the result for a smaller gap implies the result for a larger gap. Let the number of clauses in $\mathcal{C}$ be $m = \rho n$. We will sample with repetition from $\mathcal{C}$ to produce a list $L$ of clauses of size $t \rho n$, for some $t > 1$. We call a list {\it balanced} if:
\begin{enumerate}
    \item For every set $S \subseteq \mathcal{C}, |S| = s\rho n$, total number of occurrences of clauses from $S$ in $L$ is at most $s(1+\epsilon/3)t\rho n$.\label{cond:1}
    \item For every set $S \subseteq \mathcal{C}, |S| = s(1+\epsilon)\rho n$, total number of occurrences of clauses from $S$ in $L$ is at least $s(1+2\epsilon/3)t\rho n$.\label{cond:2}
\end{enumerate}

It is easy to see that the probability of sampling an unbalanced list is:
\begin{align*}
\Pr[\text{L is unbalanced}] &\leq \binom{\rho n}{s\rho n}\exp\left(-\epsilon^2st\rho n/9\right) + \binom{\rho n}{s(1+\epsilon)\rho n}\exp\left(-\epsilon^2s(1+\epsilon)t\rho n/16\right) \\
&\leq \exp(-10\rho n),
\end{align*}
when $t$ is a large enough constant depending on $s, \epsilon$.

Let $\mathcal{C}'$ be the CSP given by the set of clauses in $L$ (repeated clauses might be present in $\mathcal{C}$). If $L$ is balanced then the soundness of $\mathcal{C}'$ is $\leq s(1+\epsilon/3)$ and completeness is $\geq s(1+2\epsilon/3)$. If our list is not balanced we will reject it and output any NO instance. This can be done in polynomial time as we can check the condition~\ref{cond:1} by finding a set of clauses of size $s\rho n$ which occurs the most and checking that it occurs at most $s(1+\epsilon/3)t\rho n$ in $L$. We can similarly check condition~\ref{cond:2}.

Let $(S_i)_{i = 1}^{\abs{L}}$ be the set family given to us by the expander sampler from Lemma~\ref{lem:exp-samp} with parameters $\mathcal{ES}((1/(s^2\epsilon^2k)), s\epsilon, \abs{L})$. Consider new clauses $B_i$ such that each clause is a threshold gate, i.e. $B_i = \Thr{s(1+\epsilon/2)}(\mathcal{C}'|_{S_i})$, where $\mathcal{C}'$ denotes the vector of clauses of $L$. By the sampler property $|S_i| \leq O(k)$ and the number of $B_i$'s is equal to $|L| = t\rho n$. 

Our final CSP will be given by the set of clauses $B_i$. For the $i^{th}$ clause will find the values of all $C_j'$ such that $j \in S_i$ and then verify that their threshold value is $\geq s(1+\epsilon/2)$. Our query size is $O(k)$ as we find values for variables in $O(k)$ clauses each of them is on 3 variables.

\begin{paragraph}{\bf Soundness}
If $L$ is balanced, in the NO case the soundness is $\leq s(1+\epsilon/3)$. Then we get that,

\begin{align*}
    \Pr_i[B_i(\mathcal{C'}) = 1] &= \Pr\left[\frac{1}{\abs{S_i}}\sum_{j \in S_{i}}\mathcal{C'}_j \geq s(1+\epsilon/2)\right]\\
    &= \Pr\left[\frac{1}{\abs{S_i}}\sum_{j \in S_{i}}\mathcal{C'}_j-s(1+\epsilon/3) \geq s(\epsilon/6)\right]\\
    &\leq 1/(s^2\epsilon^2k)
\end{align*}
where the last inequality follows from the properties of expander sampler in Lemma~\ref{lem:exp-samp}. Now for large enough $k$ we get $1/(s^2\epsilon^2k) \leq 1/2$ hence starting from all NO instances gives us NO instances.

If $L$ is unbalanced we always output a NO instance.
\end{paragraph}

\begin{paragraph}{\bf Completeness}

By the property of the expander sampler in Lemma~\ref{lem:exp-samp}, the number of query sets that intersect with some query set $S_i$ are at most $O(k^2)$ for large enough $k$. As the original instance was a YES instance there exists an $x = x_c$ which satisfies $s(1+\epsilon)$ fraction of the clauses. As each clause of list $L$ is a random clause from the original set of clauses, the probability that any specific $B_i$ evaluates to $1$ is $\geq 1- \exp\left(-\Omega\left(-\epsilon^2ks\right)\right)$ by the Chernoff bound for assignment $x_c$.

As each clause of list $L$ is a random clause from the original set of clauses, we get that the random variables (randomness from choosing the list $L$, for fixed sets $S_i$) $B_i$ and $B_j$ are independent, if two query sets $S_i$ and $S_j$ do not intersect.  As calculated above, the probability that any clause fails is $\leq \exp\left(-\Omega\left(-\epsilon^2ks\right)\right)$. For large enough constants $k$, $$e\cdot O\left(k^2\right)\exp\left(-\Omega\left(-\epsilon^2ks\right)\right) < 1$$ which allows us to apply the Lov\'asz local lemma as given in Lemma~\ref{lem:lll}. This gives us that for large enough $k$,
$$\Pr_L[\wedge B_i(\mathcal{C'}) = 1] \geq (1-1/O(k^2))^{t\rho n} \geq  2^{-n/(2k)}.$$  

Taking into account the case where $L$ is unbalanced, the probability of outputting a YES instance is $\geq 2^{-n/(2k)}-2^{-10 \rho n} \geq 2^{-n/k}$ for large enough $k$.
\end{paragraph}
\end{proof}

\section{Conclusion}
The reduction in Section~\ref{sec:logn} is not useful to get perfect completeness for PCPs, while preserving their query complexity and losing some factor in the randomness complexity. When the construction is composed with query reduction, it only gives us that $\PCP_{c,s}[\log n+O(1), O(1)] \subseteq \PCP_{1,s'}[\log n +O(\log\log n),O(1)]$, which is anyway the blow-up incurred in state of the art PCPs for $\NP$~\cite{dinur}. Hence we pose the following problem:
\begin{open}
Let $c, s, s' \in (0,1)$ with $s < c$ be constants. Then is it true that,
$$\PCP_{c,s}[\log n + O(1), O(1)] \subseteq \PCP_{1,s'}(\log n + o(\log\log n), O(1))?$$
\end{open}

\paragraph{Acknowledgments} We would like to thank our advisors Madhu Sudan and Ryan Williams for helpful discussions. We are also grateful to all the reviewers, for detailed comments on the paper.

\bibliography{references}
\appendix

\section{Proof of Theorem~\ref{samp}}\label{a:b}
We first consider an expander sampler from Lemma~\ref{lem:exp-samp}. This is an expander graph $G$ over $N$ nodes, with the set family $\mathcal{ES} = (S_v)_{v \in [N]}$, where $S_v = $ set of neighbors of $v$ in $G$. To get parameters $\mathcal{S}(\epsilon,\delta/2,N)$ we can see that from the proof in~\cite{GoldreichW97}, if one takes the second eigenvalue $\lambda$ to be small enough,  $\leq \poly(\epsilon,\delta)$, then the following holds: 

For any string $x \in \{0, 1\}^N$, 
\begin{equation}\label{samp:s}
\Pr\limits_{S \sim \mathcal{ES}}[|\overline{(x|_{S})} - \overline{x}| > \epsilon] \leq \delta/2. 
\end{equation}

Now let us see, how to achieve property (2). We use the Expander Mixing Lemma, and show (proof deferred to later in this section) that if $\lambda$ is small enough ($\leq \poly(\gamma)$) then the following holds:
For all $\eta < (1-\gamma)/2$, for any string $x \in \{0, 1\}^N$, where $\overline{x} \geq 1 - \eta$, we get that,
\begin{equation}\label{samp:c}
\Pr\limits_{S \sim \mathcal{ES}}[\overline{(x|_{S})} < \gamma] \leq \eta/4.
\end{equation}

Taking the second eigenvalue $\lambda$ less than the minimum required in both proofs above, we get that both the above statements hold, for some $\lambda = O_{\epsilon,\delta,\gamma}(1)$. Note that the degree of an expander, which is also the sample complexity, is $\poly(1/\lambda) = O_{\epsilon,\delta,\gamma}(1) = C$, hence property (3) holds.

To get property (4), we arbitrarily take $N/2$ of the samples and define this as the set family $\mathcal{S}$ given by the sampler. This hurts the probabilities in equations~\ref{samp:s} and~\ref{samp:c} by a factor of at most $2$ and hence we get properties (1), (2) for the new set family.

\begin{proof}[Proof of property(2)]
Let $B \subset [N]$ be the positions of zeros in $x$ and let $C \subset [N]$ be the set of vertices that have at least $1-\gamma$-fraction of their neighbors in $B$. Notice that the vertices $S$ in $C$ are exactly those samples on which $\overline{x|_S} < 1-\gamma$, hence it is enough to bound $|C|/N = \eta'$.

We know that $|B|/N = \eta$ and let $|C|/N = \eta'$. By the Expander mixing lemma we get that,
$$\frac{|E(B,C)|}{|E(G)|} \leq \eta\eta' + \ \lambda\sqrt{\eta\eta'},$$ where $\lambda$ is the second eigenvalue of graph $G$.
But by the property of $C$, we get, $\frac{|E(B,C)|}{|E(G)|} \geq (1-\gamma)\eta'$. Combining these two we get that $\eta' \leq (\lambda^2/(1-\gamma-\eta)^2)\eta \leq (4\lambda^2/(1-\gamma)^2)\eta$. We can take $\lambda$ to be small enough in terms of $\gamma$, to get that $\eta' < \eta/4$.
\end{proof}

\end{document}